\documentclass[a4paper]{article}
\usepackage{geometry}
\usepackage{amsmath}
\usepackage{amssymb}
\usepackage{indentfirst}
\usepackage{mathdots}
\usepackage{cite}
\usepackage[colorlinks, linkcolor=blue, citecolor=blue]{hyperref}
\usepackage{booktabs}
\usepackage{enumitem}
\usepackage{mathrsfs}
\usepackage{enumitem}
\usepackage{authblk}
\usepackage{amsfonts,ntheorem}
\usepackage{float}
\usepackage{setspace}
\numberwithin{equation}{section}
\usepackage{color}

\usepackage[justification=centering]{caption}

\usepackage{graphicx}
\usepackage{subfigure}

\makeatletter

\newcommand{\Rmnum}[1]{\expandafter\@slowromancap\romannumeral #1@}
\makeatother

\newtheorem{theorem}{Theorem}[section]
\newtheorem{lemma}[theorem]{Lemma}
\newtheorem{corollary}[theorem]{Corollary}
\newtheorem{proposition}[theorem]{Proposition}

{ 
\theoremstyle{plain}
\theorembodyfont{\normalfont} 

\newtheorem{remark}[theorem]{Remark}
\newtheorem{example}[theorem]{Example}
\newtheorem{definition}[theorem]{Definition}
\newtheorem{claim}[theorem]{Claim}
}

\newenvironment{proof}{\noindent{\textbf{\emph{Proof.}}}}

\allowdisplaybreaks[4]

\begin{document}

\title{Dynamical quantum state tomography with time-dependent channels}

\author[1]{{\small{Meng Cao}} \thanks{E-mail address: mengcao@bimsa.cn}}
\author[1]{{\small{Yu Wang}} \thanks{E-mail address: wangyu@bimsa.cn}}
\affil[1]{\footnotesize{Yanqi Lake Beijing Institute of Mathematical Sciences and Applications, Beijing, 101408, China} }
\renewcommand*{\Affilfont}{\small\it}

\date{}
\maketitle

\vspace{-30pt}

{\linespread{1.4}{

\begin{abstract}
In this paper, we establish a dynamical quantum state tomography framework. Under this framework, it is feasible to obtain complete knowledge of any unknown state of a $d$-level system via only an arbitrary operator of certain types of IC-POVMs in dimension $d$.
We show that under the time-dependent average channel, we can acquire a collection of projective operators that is informationally complete (IC) and thus obtain the corresponding IC-POVMs. We show that under certain condition, it is possible to obtain infinite families of projective operators that are IC, and obtain infinite families of corresponding IC-POVMs; otherwise, the Zauner's conjecture is incorrect.
We also show how to simulate a SIC-POVM on any unknown quantum state by using the time-dependent average channel.
\end{abstract}


\vspace{6pt}
}


\section{Introduction}

Let $\mathbb{C}^{d}$ denote the $d$-dimensional vector space over the complex field $\mathbb{C}$.
Let $M_{d}(\mathbb{C})$ denote the set consisting of all $d\times d$ matrices over $\mathbb{C}$.
An operator $A$ acting on $\mathbb{C}^{d}$ is termed a \textbf{positive operator} if it meets the following conditions:

(1) $A$ is a Hermitian operator, i.e., $A=A^{\ast}$, where the transposed conjugate operator $A^{\ast}:=\overline{A^{T}}$ is called the adjoint operator of $A$;

(2) $\langle v|A|v\rangle\geq 0$ for any nonzero vector $|v\rangle\in \mathbb{C}^{d}$.

If $A$ is a positive operator, we denote it as $A\geq 0$. A positive operator-valued measure (POVM) is a general measure on a quantum system.

\begin{definition}\label{definition1.1}
\rm{(\!\!\cite{Davies1976Quantum,Busch1997Operational,Peres2006Quantum,Nielsen2000Quantum})}
Let $X$ be a set. A collection of operators $\{E_{i}\}_{i\in X}$ on $\mathbb{C}^{d}$ is called a \textbf{positive operator-valued measure (POVM)}
if it satisfies
the following two conditions:

(i) $E_{i}\geq 0$ for each $i$;

(ii) $\sum_{i\in X}E_{i}=I_{d}$.
\end{definition}

We call $\rho\in M_{d}(\mathbb{C})$ a $d$-dimensional quantum state (or, a qudit) if $\rho\geq 0$ with $\mathrm{tr}(\rho)=1$.
Consider a measurement described by a POVM $\{E_{i}\}_{i\in X}$ performed on $\rho$. The probability of obtaining outcome $i\in X$ associated with $E_{i}$ follows the Born rule \cite{Born1955Statistical}:
$p(i)=\mathrm{tr}(E_i\rho)$. Note that these probabilities $\{p(i)\}_{i\in X}$ must sum up to one, i.e., $\sum_{i\in X}p(i)=1$, which is equivalent to $\sum_{i\in X}E_{i}=I_{d}$ since $\mathrm{tr}(\rho)=1$.
If the unknown state $\rho$ can be uniquely determined from these probabilities $\{p(i)\}_{i\in X}$, we refer to the POVM $\{E_{i}\}_{i\in X}$ as an
\textbf{informationally complete POVM (IC-POVM)}.

An IC-POVM possesses the following property.

\begin{proposition}\label{proposition1.2}
\rm{(\!\!\cite[Theorem 2.4]{Ohno2014Necessary})}
\emph{Let $\{E_{i}\}_{i\in X}$ be a POVM on $\mathbb{C}^{d}$.
Then $\{E_{i}\}_{i\in X}$ is an IC-POVM if and only if $\mathrm{span}\{E_{i}\}_{i\in X}=M_{d}(\mathbb{C})$.
}
\end{proposition}

Occasionally, we regard the ``if'' portion of this proposition as defining an IC-POVM 
(see, e.g., \cite[Definition 1]{Czerwinski2021Quantum}).
It is evident that the cardinality of an IC-POVM on $\mathbb{C}^{d}$ is no less than $d^{2}$.
A notably attractive family of IC-POVMs on $\mathbb{C}^{d}$ are the symmetric informationally complete POVMs (SIC-POVMs), which consist precisely of $d^{2}$ elements.

\begin{definition}\label{definition1.3}
\rm{(\!\!\cite[pp.2171-2172]{Renes2004Symmetric})}
A POVM $\{E_{1},E_{2},\ldots,E_{d^{2}}\}$ on $\mathbb{C}^{d}$ is called a \textbf{symmetric informationally complete POVM (SIC-POVM)} if it satisfies
the following two coditions:

(i) $E_{i}=\frac{1}{d}|u_{i}\rangle\langle u_{i}|$, where $|u_{i}\rangle$ is a normalized vector in $\mathbb{C}^{d}$ for each $i=1,2,\ldots,d^{2}$;

(ii) $|\langle u_{i}|u_{j}\rangle|^{2}=d^{2}\mathrm{tr}(E_{i}E_{j})=\frac{1}{d+1}$ for all $1\leq i\neq j\leq d^{2}$.
\end{definition}

The SIC-POVM on $\mathbb{C}^{d}$ can be understood as $d^{2}$ pairwise equiangular complex lines, representing $d^{2}$ one-dimensional subspaces $\mathbb{C}|u_{1}\rangle,\mathbb{C}|u_{2}\rangle,\ldots,\mathbb{C}|u_{d^{2}}\rangle$ in $\mathbb{C}^{d}$, as initially explored by Lemmens and Seidal in \cite{Lemmens1973Equiangular}. 
Subsequently, research on SIC-POVMs and their variations has experienced rapid advancement (see, for instance, \cite{Appleby2015Group,Appleby2011The,Delsarte1991Bounds,Hoggar1998Line,Konig1994Norms,Konig1999Cubature,Strohmer2003Grassmannian,
Zauner1999Quantum,Kopp2021SIC-POVMs,Zhu2010SIC,Zhu2018Universally,Tavakoli2020Compounds,Petz2014Conditional,Gour2014Construction,Geng2021What}). 
SIC-POVMs are intricately linked to various fields such as quantum state tomography \cite{Czerwinski2021Quantum,Caves2002Unknown,Scott2006Tight},
quantum cryptography \cite{Fuchs2003Squeezing,Fuchs2004On}, design theory \cite{Renes2004Symmetric,Zauner2011Quantum,Klappenecker2005Mutually},
and frame theory \cite{Fickus2021Mutually,Fickus2018Tremain,Cahill2018Constructions,Magsino2019Biangular}.

According to Zauner's conjecture outlined in his 1999 PhD thesis \cite{Zauner1999Quantum}, SIC-POVMs are posited to exist across all finite dimensions $d\geq 2$.
Analytical constructions of SIC-POVMs on $\mathbb{C}^{d}$ have been provided for dimensions $d=2$-$24$, $28$, $30$, $31$, $35$, $37$, $39$, $43$, $48$, $124$, and $323$
(see \cite{Zauner1999Quantum,Appleby2005Symmetric,Appleby2014Systems,Appleby2012The,Appleby2018Constructing,Grassl2009OnSIC,
Grassl2005Tomography,Grassl2008Computing,Scott2010Symmetric,Grassl2017FibonacciLucas}).
Furthermore, numerical solutions of SIC-POVMs obtained through computational methods have been found for all dimensions up to $d=151$,
as well as several other dimensions up to $d=844$
(see \cite{Renes2004Symmetric,Scott2010Symmetric,Grassl2017FibonacciLucas,Fuchs2017The,Scott2017SICs}).
Although it is widely speculated that SIC-POVMs exist in every finite dimension, a formal proof of this assertion remains elusive.
Moreover, no infinite families of SIC-POVMs have been previously constructed, and it is uncertain whether SIC-POVMs on $\mathbb{C}^{d}$
exist for infinitely many dimensions $d$.

As we know, completely positive and trace-preserving (CPTP) maps serve as a means to describe transformations within quantum systems.
To record changes in the quantum system over time, it is convenient to use dynamical maps that are time-dependent CPTP maps.
As a kind of dynamical maps, the random unitary dynamics (RUD) is used to characterize evolution of states of quantum systems.
Let $\rho(0)$ represent an arbitrary unknown state of a $d$-level quantum system.
In quantum theory, if we can reconstruct the unknown state $\rho(0)$ through a set of IC measurements on identical copies of
$\rho(0)$, then we call this process quantum state tomography \cite{Paris2004Quantum}.
To search optimal approaches of quantum state tomography, Czerwinski \cite{Czerwinski2021Quantum} recently considered the
following qudit dynamics whose evolution is described by RUD:
\begin{equation*}
\rho(t)=\sum_{i\in X}\mu_{i}(t)U_{i}\rho(0)U_{i}^{\dag},
\end{equation*}
where the unknown state $\rho(0)$ is regarded as the initial state of the quantum system, $\{U_{i}\}_{i\in X}$ is a collection of unitary matrices,
and $\{\mu_{i}(t)\}_{i\in X}$ is a time-continuous probability distribution, i.e.,
$\sum_{i\in X}\mu_{i}(t)=1$ with $\mu_{i}(t)\geq 0$ for all $i\in X$ and any $t\geq 0$.
Czerwinski showed that by starting with an incomplete set of certain IC-POVMs (resp. SIC-POVMs) on $\mathbb{C}^{2}$ (resp. $\mathbb{C}^{3}$),
it is feasible to obtain complete information for the reconstruction of qubits (resp. qutrits) via multiple measurements.
More precisely, it suffices to initially have one measurement operator for certain dynamical qubits and qutrits.
These results suggest that qubits tomography and qutrits tomography are feasible even when the measurement potential is limited.

Inspired by the work in \cite{Czerwinski2021Quantum}, in Section \ref{section2}, we establish a dynamical quantum state tomography framework.
To be specific, we show that for any unknown qudit $\rho(0)$ of the $d$-level system, by regarding it as the initial state of the following
evolution subject to a type of RUD:
$$\rho(t)=\sum_{i\in X}\mu_{i}(t)H_{i}\rho(0)H_{i}^{\dag},$$
it is feasible to extract complete knowledge of $\rho(0)$ via only an arbitrary operator (matrix) $Q_{j}$ ($j\in X$) of certain types of IC-POVMs $Q:=\{Q_{i}\}_{i\in X}$ (see Theorem \ref{theorem4.4}).
Here, the unitary operators $\{H_{i}\}_{i\in X}$ are actually a collection of quasi-Householder operators, where each operator $H_{i}$ is determined
by operators $Q_{i}$ and $Q_{j}$.
This framework provides a general approach to reconstruct the unknown state $\rho(0)$
for qudits tomography, extending the cases of qubits tomography and qutrits tomography shown in \cite{Czerwinski2021Quantum}.
Since an arbitrary operator $Q_{j}$ of the IC-POVM $Q:=\{Q_{i}\}_{i\in X}$ in our framework is sufficient to reconstruct any unknown qudit $\rho(0)$,
one needs to prepare one experimental setup and then repeat the same kind of measurement at $|X|$ distinct time instants, which is more convenient than preparing a number of distinct experimental setups.
In brief, our framework gives a potential manner to reduce the number of distinct measurement setups needed for state reconstruction.
From this point of view, our framework is very efficient for quantum state tomography.
Besides, since the quasi-Householder operators $\{H_{i}\}_{i\in X}$ appearing in the framework are not unique,
and the time-continuous probability distribution $\{\mu_{i}(t)\}_{i\in X}$ can be chosen in a relatively free manner,
we are able to acquire many different types of dynamical qudits $\rho(t)$ that are all available for reconstruction of the unknown state $\rho(0)$.
From this perspective, the framework is very flexible.

In Section \ref{section3}, we consider another question: which families of IC-POVMs can be obtained from certain dynamical process?
To solve it, we define the time-dependent average channel as follows:
\begin{equation*}
\Psi_t(\rho)= \sum_{i=0}^{d^{2}-1} \mu _{i}(t) M_{i}\rho_0 M_{i}^{\dag},
\end{equation*}
where $\rho_0$ is an unknown quantum state,
$\mu _{0}(t)=\frac{1-\lambda_{0}(t)+d^{2}\sum_{i=0}^{d^{2}-1}\lambda_{i}(t)}{d^{4}}$,
$\mu_{i}(t)=\frac{1-\lambda_{0}(t)+d^{2}(1-\lambda_{i}(t))}{d^{4}}$ for $i=1,\ldots,d^{2}-1$, and
$M_{\alpha}:=M_{jk}=X^jZ^k$ are the Weyl-Heisenberg bases with $j,k\in \{0,\cdots,d-1\}$, $\alpha=jd+k$,
$X=\sum_{j=0}^{d-1}|j+1\rangle\langle j|$ and $Z=\sum_{j=0}^{d-1}\omega^j|j\rangle\langle j|$.
We show that under the time-dependent average channel $\Psi_t(\rho)$, we can acquire a collection of projective operators that is informationally complete (IC) and thus obtain the corresponding IC-POVM (see Theorem \ref{theorem5.1}). Moreover, we show that under certain condition, it is possible to obtain infinite families of projective operators that are IC, and obtain infinite families of corresponding IC-POVMs (see Theorem \ref{theorem5.4} and Corollary \ref{corollary5.5}); otherwise, the Zauner's conjecture is incorrect.

We show that for any unknown quantum state $\rho_{0}$ which is evolved with the time-dependent average channel $\Psi_t(\rho)$, only one projective measurement $\{|\phi\rangle\langle\phi|,I-|\phi\rangle\langle\phi|\}$ in $d^2$ different time instants could be enough to reconstruct $\rho_{0}$ (see Eq. (\ref{eq5.7})), instead of the common $d^2$ different measurements. As an application, we show that we can simulate a SIC-POVM on any unknown quantum state by using the time-dependent average channel $\Psi_t(\rho)$ (see Claim \ref{claim5.7}).

This paper is organized as follows. In Section \ref{section2}, we establish an explicit framework on how to reconstruct an unknown qudit via only an arbitrary operator of certain types of IC-POVMs (see Theorem \ref{theorem4.4}).
In Section \ref{section3}, we study the IC-POVMs under the time-dependent average channel
and show how to simulate a SIC-POVM on any unknown quantum state by this channel.
Section \ref{section4} gives the concluding remarks of this paper.

\section{Qudits tomography related to IC-POVMs}\label{section2}
In quantum theory, a POVM can be used to describe the outcome statistics of a quantum measurement.
Suppose the density matrix $\rho$ is an unknown state of a quantum system.
If we can reconstruct $\rho$ through a set of measurements on identical copies of
$\rho$, then this process is called \textbf{quantum state tomography} \cite{Paris2004Quantum}.
IC-POVMs play a central role in the context of quantum state tomography.

In this section, we will establish an explicit framework for qudits tomography by a type of random unitary dynamics (RUD).
Under this framework, the problem of reconstructing an unknown qudit via all of the operators of an IC-POVM can be reduced to that of measuring a time-dependent qudit at several distinct time instants via only an arbitrary operator of certain types of IC-POVMs.
This provides a general, efficient and flexible approach to reconstruct unknown qudits in quantum state tomography.

First, let us recall the concept of random unitary dynamics (RUD) \cite{Chru2013Nonmar,Czerwinski2021Quantum}.

\begin{definition}\label{definition4.1}
\rm{(\!\!\cite[Definition 5]{Czerwinski2021Quantum})}
Let $d$ be a positive integer and let $X$ be a set. We call a dynamical map $\Omega_{t}: M_{d}(\mathbb{C})\rightarrow M_{d}(\mathbb{C})$
\textbf{random unitary dynamics (RUD)} if it has the form:
\begin{equation*}
\Omega_{t}[W]=\sum_{i\in X}\mu_{i}(t)U_{i}WU_{i}^{\dag},
\end{equation*}
where $\Omega_{t}$ ($t\geq 0$) is a completely positive and trace-preserving (CPTP) map with $\Omega_{0}=\mathbb{I}$,
$\{U_{i}\}_{i\in X}$ is a collection of unitary matrices, and $\{\mu_{i}(t)\}_{i\in X}$ is a
time-continuous probability distribution, i.e.,
\begin{equation*}
\sum_{i\in X}\mu_{i}(t)=1
\end{equation*}
with $\mu_{i}(t)\geq 0$ for all $i\in X$ and $t\geq 0$.
\end{definition}

For a dynamical qudit $\rho(t)$, the probability of obtaining outcome $k\in X$ related to the operator $E_{k}$ of an IC-POVM
$\{E_{i}\}_{i\in X}$ is given by the Born rule:
\begin{equation*}
\mathrm{Prob}(t^{(k)})=\mathrm{tr}(E_{k}\rho(t)).
\end{equation*}

Suppose that $\rho(0)$ is an arbitrary unknown state (i.e., a qudit) of a $d$-level system. Then, based on Definition \ref{definition4.1},
one can consider the following qudit dynamics whose evolution is described by RUD:
{\setlength\abovedisplayskip{0.2cm}
\setlength\belowdisplayskip{0.2cm}
\begin{equation*}
\rho(t)=\Omega_{t}[\rho(0)]=\sum_{i\in X}\mu_{i}(t)U_{i}\rho(0)U_{i}^{\dag},
\end{equation*}
where the unknown state $\rho(0)$ is regarded as an initial state of the system.}

\vspace{5pt}

To obtain the main result of this section, let us first give the following lemma.

\begin{lemma}\label{lemma4.2}
Let $\{P_{i}\}_{i\in X}$ be a collection of subnormalized projectors on $\mathbb{C}^{d}$
with $P_{i}=p_{i}|a_{i}\rangle\langle a_{i}|$, where $p_{i}>0$ and $|a_{i}\rangle\in \mathbb{C}^{d}$ is a normalized vector for each $i\in X$.
Suppose $\mathrm{span}\{P_{i}\}_{i\in X}=M_{d}(\mathbb{C})$. Denote $P=\sum_{i\in X}P_{i}$.
Then, $P$ is positive definite.
\end{lemma}

\begin{proof}
First, it is not difficult to obtain $\mathbb{C}^{d}=\mathrm{span}\{|a_{i}\rangle\}_{i\in X}$.
For any nonzero vector $|y\rangle\in \mathbb{C}^{d}$, we have that
\begin{equation}\label{eq4.1}
\langle y|P|y\rangle=\sum_{i\in X}p_{i}|\langle y|a_{i}\rangle|^{2}\geq 0.
\end{equation}

Let us prove that $\langle y|P|y\rangle>0$. Assume that there exists $|y_{0}\rangle\neq\mathbf{0}$ such that
$\langle y_{0}|P|y_{0}\rangle=0$, then Eq. (\ref{eq4.1}) implies that $\langle y_{0}|a_{i}\rangle=0$ for each $i\in X$.
Hence, $|y_{0}\rangle\in \big(\mathrm{span}\{|a_{i}\rangle\}_{i\in X}\big)^{\perp}=(\mathbb{C}^{d})^{\perp}=\{\mathbf{0}\}$,
which leads to a contradiction. Therefore, $\langle y|P|y\rangle>0$ for any $|y\rangle\neq\mathbf{0}$, equivalently, $P$ is positive definite. $\hfill\square$
\end{proof}

\vspace{6pt}

Clearly, the cardinality of the set $X$ in Lemma \ref{lemma4.2} is at least $d^{2}$.
Since the operator $P$ defined in Lemma \ref{lemma4.2} is positive definite, we can obtain the uniquely determined positive definite operator $P^{-\frac{1}{2}}$. Following this fact, we obtain the following lemma.

\begin{lemma}\label{lemma4.3}
Let $\{P_{i}\}_{i\in X}$ be a collection of subnormalized projectors on $\mathbb{C}^{d}$
with $P_{i}=p_{i}|a_{i}\rangle\langle a_{i}|$, where $p_{i}>0$ and $|a_{i}\rangle\in \mathbb{C}^{d}$ is a normalized vector for each $i\in X$.
Suppose $\mathrm{span}\{P_{i}\}_{i\in X}=M_{d}(\mathbb{C})$.
Denote $P=\sum_{i\in X}P_{i}$ and $Q_{i}=P^{-\frac{1}{2}}P_{i}P^{-\frac{1}{2}}$ for each $i\in X$.
Then, $\mathrm{span}\{Q_{i}\}_{i\in X}=M_{d}(\mathbb{C})$.
\end{lemma}

\begin{proof}
It suffices to prove that $M_{d}(\mathbb{C})\subseteq\mathrm{span}\{Q_{i}\}_{i\in X}$.
For any operator $A\in M_{d}(\mathbb{C})$, we write it as
\begin{equation}\label{eq4.2}
A=P^{-\frac{1}{2}}(P^{\frac{1}{2}}AP^{\frac{1}{2}})P^{-\frac{1}{2}}.
\end{equation}

{\setlength\abovedisplayskip{0.25cm}
\setlength\belowdisplayskip{0.25cm}
Since $\mathrm{span}\{P_{i}\}_{i\in X}=M_{d}(\mathbb{C})$, then there exist $\{y_{i}\}_{i\in X}\in\mathbb{C}$ such that
$P^{\frac{1}{2}}AP^{\frac{1}{2}}=\sum_{i\in X}y_{i}P_{i}$. Inserting this into Eq. (\ref{eq4.2}), we have
\begin{equation*}
A=\sum_{i\in X}y_{i}P^{-\frac{1}{2}}P_{i}P^{-\frac{1}{2}}=\sum_{i\in X}y_{i}Q_{i}\in\mathrm{span}\{Q_{i}\}_{i\in X}.
\end{equation*}
Therefore, $M_{d}(\mathbb{C})\subseteq\mathrm{span}\{Q_{i}\}_{i\in X}$,
and thus $\mathrm{span}\{Q_{i}\}_{i\in X}=M_{d}(\mathbb{C})$.} $\hfill\square$
\end{proof}

\vspace{8pt}

We can verify that the set $\{Q_{i}\}_{i\in X}$, with $Q_{i}=P^{-\frac{1}{2}}P_{i}P^{-\frac{1}{2}}$ defined in Lemma \ref{lemma4.3},
forms an IC-POVM because $\sum_{i\in X}Q_{i}$ is an identity operator and $\mathrm{span}\{Q_{i}\}_{i\in X}=M_{d}(\mathbb{C})$.

\vspace{6pt}

In the following theorem, we are ready to show that for such an IC-POVM $Q:=\{Q_{i}\}_{i\in X}$ and any unknown qudit $\rho(0)$ of a $d$-level system,
there exists a time-dependent qudit (not unique) subject to RUD whose unitary matrices are a kind of quasi-Householder matrices,
such that one can extract complete knowledge about $\rho(0)$ provided the probability of an arbitrary outcome $j\in X$ related to $Q_{j}$
is measured at $|X|$ distinct time instants.

\begin{theorem}\label{theorem4.4}
Let $X=\{1,2,\ldots,x\}$ with $x\geq d^{2}$.
Let $\{P_{i}\}_{i\in X}$ be a collection of subnormalized projectors on $\mathbb{C}^{d}$
with $P_{i}=p_{i}|a_{i}\rangle\langle a_{i}|$, where $p_{i}>0$ and $|a_{i}\rangle\in \mathbb{C}^{d}$ is a normalized vector for each $i\in X$.
Suppose $\mathrm{span}\{P_{i}\}_{i\in X}=M_{d}(\mathbb{C})$.
Denote $P=\sum_{i\in X}P_{i}$, and $Q=\{Q_{i}\}_{i\in X}$ with $Q_{i}=P^{-\frac{1}{2}}P_{i}P^{-\frac{1}{2}}$ for each $i\in X$.
Then, for any unknown qudit $\rho(0)\in M_{d}(\mathbb{C})$, there exists a time-dependent qudit
\begin{equation}\label{eq4.3}
\rho(t)=\sum_{i\in X}\mu_{i}(t)H_{i}\rho(0)H_{i}^{\dag},
\end{equation}
where $\{H_{i}\}_{i\in X}$ is a collection of quasi-Householder matrices and $\{\mu_{i}(t)\}_{i\in X}$ is a time-continuous probability distribution,
such that one can extract complete knowledge about the unknown state $\rho(0)$
provided the probability of an arbitrary outcome $j\in X$ related to $Q_{j}$ is measured at $x$ distinct time instants.
\end{theorem}

\begin{proof}
We can write $Q_{i}$ as
\begin{equation}\label{eq4.4}
Q_{i}=p_{i}P^{-\frac{1}{2}}|a_{i}\rangle\langle a_{i}|P^{-\frac{1}{2}}, \ i\in X.
\end{equation}

Note that $P^{-\frac{1}{2}}|a_{i}\rangle$ can be easily rewritten as a multiple of certain normalized vector:
\begin{equation*}
P^{-\frac{1}{2}}|a_{i}\rangle=\lambda_{i}|b_{i}\rangle, \ i\in X,
\end{equation*}
{\setlength\abovedisplayskip{0.25cm}
\setlength\belowdisplayskip{0.25cm}
where $\lambda_{i}\in \mathbb{C}\backslash\{0\}$, and $|b_{i}\rangle\in \mathbb{C}^{d}$ is a normalized vector.
To be specific, suppose $P^{-\frac{1}{2}}|a_{i}\rangle=(p_{i,1},\ldots,p_{i,d})^{T}\in \mathbb{C}^{d}$, then
\begin{align}\label{eq4.5}
P^{-\frac{1}{2}}|a_{i}\rangle&=\frac{\sqrt{\sum_{s=1}^{d}|p_{i,s}|^{2}}}{z_{i}}\Bigg(\frac{p_{i,1}z_{i}}{\sqrt{\sum_{s=1}^{d}|p_{i,s}|^{2}}},\ldots,
\frac{p_{i,d}z_{i}}{\sqrt{\sum_{s=1}^{d}|p_{i,s}|^{2}}}\Bigg)^{T}\nonumber\\
&\triangleq \lambda_{i}|b_{i}\rangle,
\end{align}}
where $\lambda_{i}=\frac{\sqrt{\sum_{s=1}^{d}|p_{i,s}|^{2}}}{z_{i}}$
and $|b_{i}\rangle=\Big(\frac{p_{i,1}z_{i}}{\sqrt{\sum_{s=1}^{d}|p_{i,s}|^{2}}},\ldots,
\frac{p_{i,d}z_{i}}{\sqrt{\sum_{s=1}^{d}|p_{i,s}|^{2}}}\Big)^{T}$
in which $z_{i}$ is an arbitrary complex number with $|z_{i}|=1$.

Inserting Eq. (\ref{eq4.5}) into Eq. (\ref{eq4.4}), we have
\begin{equation}\label{eq4.6}
Q_{i}=p_{i}|\lambda_{i}|^{2}|b_{i}\rangle\langle b_{i}|, \ i\in X.
\end{equation}

Let us choose an arbitrary index $j$ from the set $X$. Then, for the corresponding vector $|b_{j}\rangle$ and all vectors $|b_{i}\rangle$, where
$i\in X$, we can always find a collection of quasi-Householder matrices $\{\widehat{H_{i}}\}_{i\in X}$, which are unitary matrices, such that
\begin{equation}\label{eq4.7}
\widehat{H_{i}}|b_{j}\rangle=|b_{i}\rangle, \ i\in X.
\end{equation}

This is because, first, there exists $\eta_{i}\in\mathbb{C}$ such that
\begin{equation}\label{eq4.8}
\overline{\eta_{i}}\langle b_{i}|b_{j}\rangle\in\mathbb{R}\ \ \mathrm{with}\ \ |\eta_{i}|=1
\end{equation}
for each $i\in X$, where $\overline{\eta_{i}}$ denotes the complex conjugate of $\eta_{i}$; afterward, one has the following two cases:

{\bfseries Case 1:} When $|b_{j}\rangle=\eta_{i}|b_{i}\rangle$, it is not difficult to find a normalized vector $|\widetilde{u_{i}}\rangle\in\mathbb{C}^{d}$ such that
\begin{equation}\label{eq4.9}
\langle \widetilde{u_{i}}|b_{j}\rangle=0.
\end{equation}
Then, one can verify that
\begin{equation*}
\overline{\eta_{i}}\big(I_{d}-2|\widetilde{u_{i}}\rangle\langle \widetilde{u_{i}}|\big)|b_{j}\rangle=|b_{i}\rangle.
\end{equation*}
In this case, take
\begin{equation}\label{eq4.10}
\widehat{H_{i}}=\overline{\eta_{i}}\big(I_{d}-2|\widetilde{u_{i}}\rangle\langle \widetilde{u_{i}}|\big).
\end{equation}
Here, $I_{d}-2|\widetilde{u_{i}}\rangle\langle \widetilde{u_{i}}|$ is a Householder matrix.

\vspace{6pt}

{\bfseries Case 2:} When $|b_{j}\rangle\neq\eta_{i}|b_{i}\rangle$, denote
\begin{equation}\label{eq4.11}
|\widetilde{v_{i}}\rangle=\frac{|b_{j}\rangle-\eta_{i}|b_{i}\rangle}{|||b_{j}\rangle-\eta_{i}|b_{i}\rangle||_{2}}.
\end{equation}
Then, one can verify that
\begin{equation*}
\overline{\eta_{i}}\big(I_{d}-2|\widetilde{v_{i}}\rangle\langle \widetilde{v_{i}}|\big)|b_{j}\rangle=|b_{i}\rangle.
\end{equation*}
In this case, take
\begin{equation}\label{eq4.12}
\widehat{H_{i}}=\overline{\eta_{i}}\big(I_{d}-2|\widetilde{v_{i}}\rangle\langle \widetilde{v_{i}}|\big).
\end{equation}
Here, $I_{d}-2|\widetilde{v_{i}}\rangle\langle \widetilde{v_{i}}|$ is a Householder matrix.

\vspace{6pt}

From {\bfseries Cases 1} and {\bfseries 2}, we can always find the quasi-Householder matrix $\widehat{H_{i}}$ such that Eq. (\ref{eq4.7}) holds.
Substituting Eq. (\ref{eq4.7}) into Eq. (\ref{eq4.6}) we obtain
\begin{align}\label{eq4.13}
Q_{i}&=p_{i}|\lambda_{i}|^{2}\widehat{H_{i}}|b_{j}\rangle\langle b_{j}|\widehat{H_{i}}^{\dag}\nonumber\\
     &=\frac{p_{i}|\lambda_{i}|^{2}}{p_{j}|\lambda_{j}|^{2}}\widehat{H_{i}}(p_{j}|\lambda_{j}|^{2}|b_{j}\rangle\langle b_{j}|)\widehat{H_{i}}^{\dag}\nonumber\\
     &=\frac{p_{i}|\lambda_{i}|^{2}}{p_{j}|\lambda_{j}|^{2}}\widehat{H_{i}}Q_{j}\widehat{H_{i}}^{\dag}.
\end{align}

Let $t_{1},t_{2},\ldots,t_{x}$ ($x\geq d^{2}$) be $x$ distinct time instants. Then, one can select a time-continuous probability distribution
$\{\mu_{i}(t)\}_{i\in X}$ such that $\mathrm{det}([\mu_{j}(t_{i})]_{i,j=1}^{x})\neq 0$.

Define $H_{i}=\widehat{H_{i}}^{\dag}$, then $H_{i}$ is a quasi-Householder matrix as well.
Now, consider the following dynamical qudit whose evolution is described by RUD with quasi-Householder matrices:
\begin{equation*}
\rho(t)=\sum_{i\in X}\mu_{i}(t)H_{i}\rho(0)H_{i}^{\dag}.
\end{equation*}
Then, the probability of obtaining outcome $j$ related to $Q_{j}$ is given by Born rule:
\begin{align}\label{eq4.14}
\mathrm{Prob}(t^{(j)})&=\mathrm{tr}(Q_{j}\rho(t))\nonumber\\
&=\mathrm{tr}\Big(Q_{j}\big(\sum_{i\in X}\mu_{i}(t)H_{i}\rho(0)H_{i}^{\dag}\big)\Big)\nonumber\\
&=\sum_{i\in X}\mu_{i}(t)\mathrm{tr}(H_{i}^{\dag}Q_{j}H_{i}\rho(0)).
\end{align}

{\setlength\abovedisplayskip{0.2cm}
\setlength\belowdisplayskip{0.2cm}
For convenience, denote $\widetilde{p_{i}}=\frac{p_{i}|\lambda_{i}|^{2}}{p_{j}|\lambda_{j}|^{2}}$, then Eq. (\ref{eq4.13}) can be rewritten as
\begin{equation*}
Q_{i}=\widetilde{p_{i}}\widehat{H_{i}}Q_{j}\widehat{H_{i}}^{\dag}=\widetilde{p_{i}}H_{i}^{\dag}Q_{j}H_{i}, \ i\in X.
\end{equation*}

Taking $t=t_{1},t_{2},\ldots,t_{x}$ in Eq. (\ref{eq4.14}), we have
\begin{align*}
\left[
\begin{array}{c}
\mathrm{Prob}\big(t_{1}^{(j)}\big)\\
\mathrm{Prob}\big(t_{2}^{(j)}\big)\\
\vdots\\
\mathrm{Prob}\big(t_{x}^{(j)}\big)\\
\end{array}\right]&=\left[
\begin{array}{cccc}
\frac{\mu_{1}(t_{1})}{\widetilde{p_{1}}}&\frac{\mu_{2}(t_{1})}{\widetilde{p_{2}}}&\cdots&\frac{\mu_{x}(t_{1})}{\widetilde{p_{x}}}\\
\frac{\mu_{1}(t_{2})}{\widetilde{p_{1}}}&\frac{\mu_{2}(t_{2})}{\widetilde{p_{2}}}&\cdots&\frac{\mu_{x}(t_{2})}{\widetilde{p_{x}}}\\
\vdots&\vdots&\ddots&\vdots\\
\frac{\mu_{1}(t_{x})}{\widetilde{p_{1}}}&\frac{\mu_{2}(t_{x})}{\widetilde{p_{2}}}&\cdots&\frac{\mu_{x}(t_{x})}{\widetilde{p_{x}}}\\
\end{array}\right]\left[
\begin{array}{c}
\mathrm{tr}(\widetilde{p_{1}}H_{1}^{\dag}Q_{j}H_{1}\rho(0))\\
\mathrm{tr}(\widetilde{p_{2}}H_{2}^{\dag}Q_{j}H_{2}\rho(0))\\
\vdots\\
\mathrm{tr}(\widetilde{p_{x}}H_{x}^{\dag}Q_{j}H_{x}\rho(0))\\
\end{array}\right]\\
\\[-2.5ex]
&\triangleq K\left[
\begin{array}{c}
\mathrm{tr}(Q_{1}\rho(0))\\
\mathrm{tr}(Q_{2}\rho(0))\\
\vdots\\
\mathrm{tr}(Q_{x}\rho(0))\\
\end{array}\right],
\end{align*}
where $K:=\big[\frac{\mu_{j}(t_{i})}{\widetilde{p_{j}}}\big]_{i,j=1}^{x}$. Clearly, the matrix $K$ is invertible since
\begin{equation}\label{eq4.15}
\mathrm{det}(K)=\Big(\prod_{i=1}^{x}\widetilde{p_{i}}^{-1}\Big)\mathrm{det}([\mu_{j}(t_{i})]_{i,j=1}^{x})\neq 0.
\end{equation}

Therefore, once the probabilities $\big\{\mathrm{Prob}\big(t_{i}^{(j)}\big)\big\}_{i\in X}$ of an arbitrary outcome $j\in X$
are experimentally known, we can directly obtain the values of the set $\{\mathrm{tr}(Q_{i}\rho(0))\}_{i\in X}$, which implies that
the complete knowledge about the unknown state $\rho(0)$ can be extracted since the set $\{Q_{i}\}_{i\in X}$ forms an IC-POVM.} $\hfill\square$
\end{proof}

\begin{remark}\label{remark4.5}
The framework in Theorem \ref{theorem4.4} has the following advantages and characteristics:
\begin{itemize}
\item It provides a general approach that can be performed for qudits tomography,
extending the cases of qubits tomography and qutrits tomography studied in \cite{Czerwinski2021Quantum}.

\item It gives a potential manner to reduce the number of distinct measurement setups needed for state reconstruction.
To be specific, for qudits tomography by IC-POVMs on $\mathbb{C}^{d}$, the standard approach usually needs to realize at least $d^{2}$ distinct measurements associated with each operator of an IC-POVM. By contrast, our framework shows that an arbitrary operator $Q_{j}$ of the IC-POVM $Q:=\{Q_{i}\}_{i\in X}$ in Theorem \ref{theorem4.4} is sufficient to reconstruct any unknown state $\rho(0)$ in qudits tomography.
This implies that one needs to prepare one experimental setup and then repeat the same kind of measurement at $|X|$ distinct time instants,
which is more convenient than preparing a number of distinct experimental setups. From this point of view, the framework is very efficient for
qudits tomography.

\item Note that the vector $|b_{i}\rangle$ in Eq. (\ref{eq4.5}), the complex number $\eta_{i}$ in Eq. (\ref{eq4.8}), the vector
$|\widetilde{u_{i}}\rangle$ in Eq. (\ref{eq4.9}) and the vector $|\widetilde{v_{i}}\rangle$ in Eq. (\ref{eq4.11}) are not unique.
Consequently, the corresponding quasi-Householder matrices $\widehat{H_{i}}$ in Eqs. (\ref{eq4.10}) and (\ref{eq4.12}) are not unique,
and so are their conjugate transpose matrices $H_{i}$.
Besides, the time-continuous probability distribution $\{\mu_{i}(t)\}_{i\in X}$ can be chosen in a relatively free manner provided
$\mathrm{det}([\mu_{j}(t_{i})]_{i,j=1}^{x})\neq 0$. All of these facts allow us to obtain many different types of dynamical qudits $\rho(t)$
(see Eq. (\ref{eq4.3})), which are all available for the reconstruction of the unknown state $\rho(0)$.
From this perspective, the framework is very flexible.
\end{itemize}
\end{remark}

To illustrate the framework in Theorem \ref{theorem4.4}, let us give the following example.

\begin{example}\label{example4.8}
Consider $j=7$ in Theorem \ref{theorem4.4}. Let $\mathbb{F}_{3}=\{a_{1},a_{2},a_{3}\}$,
where $a_{1}=0$, $a_{2}=1$ and $a_{3}=-1$.
Put $f(x)=x^{2}-x+1\in\mathbb{F}_{3}[x]$, then $f(a_{1})=1$, $f(a_{2})=1$ and $f(a_{3})=0$.
Let $\chi$ be a nontrivial additive character with $\chi(x)=e^{\frac{2\pi i}{3}x}$ for each $x\in \mathbb{F}_{3}$.
Then, $\chi(a_{1})=1$, $\chi(a_{2})=e^{\frac{2\pi i}{3}}$ and $\chi(a_{3})=e^{-\frac{2\pi i}{3}}$. Define
\begin{equation*}
|v_{1}\rangle:=\frac{1}{\sqrt{3}}\big(\chi(a_{i})\big)_{i=1,2,3}^{T}
=\frac{1}{\sqrt{3}}\big(1,e^{\frac{2\pi i}{3}},e^{-\frac{2\pi i}{3}}\big)^{T},
\end{equation*}
\begin{equation*}
|v_{2}\rangle:=\frac{1}{\sqrt{3}}\big(\chi(-a_{i})\big)_{i=1,2,3}^{T}
=\frac{1}{\sqrt{3}}\big(1,e^{-\frac{2\pi i}{3}},e^{\frac{2\pi i}{3}}\big)^{T},
\end{equation*}
\begin{equation*}
|v_{3}\rangle:=\frac{1}{\sqrt{3}}\big(\chi(f(a_{i})+a_{i})\big)_{i=1,2,3}^{T}
=\frac{1}{\sqrt{3}}\big(e^{\frac{2\pi i}{3}},e^{-\frac{2\pi i}{3}},e^{-\frac{2\pi i}{3}}\big)^{T},
\end{equation*}
\begin{equation*}
|v_{4}\rangle:=\frac{1}{\sqrt{3}}\big(\chi(f(a_{i})-a_{i})\big)_{i=1,2,3}^{T}
=\frac{1}{\sqrt{3}}\big(e^{\frac{2\pi i}{3}},1,e^{\frac{2\pi i}{3}}\big)^{T},
\end{equation*}
\begin{equation*}
|v_{5}\rangle:=\frac{1}{\sqrt{3}}\big(\chi(-f(a_{i})+a_{i})\big)_{i=1,2,3}^{T}
=\frac{1}{\sqrt{3}}\big(e^{-\frac{2\pi i}{3}},1,e^{-\frac{2\pi i}{3}}\big)^{T},
\end{equation*}
\begin{equation*}
|v_{6}\rangle:=\frac{1}{\sqrt{3}}\big(\chi(-f(a_{i})-a_{i})\big)_{i=1,2,3}^{T}
=\frac{1}{\sqrt{3}}\big(e^{-\frac{2\pi i}{3}},e^{\frac{2\pi i}{3}},e^{\frac{2\pi i}{3}}\big)^{T},
\end{equation*}
\begin{equation*}
|v_{7}\rangle:=(1,0,0)^{T},\ |v_{8}\rangle:=(0,1,0)^{T},\ |v_{9}\rangle:=(0,0,1)^{T}.
\end{equation*}

Denote $E=\sum_{i=1}^{9}E_{i}$ with $E_{i}:=\frac{1}{3}|v_{i}\rangle\langle v_{i}|$ for $i=1,2,\ldots,9$.
One can verify that $E_{1},E_{2},\ldots,E_{9}$ are linearly independent as elements of $M_{3}(\mathbb{C})$,
Besides, it follows that
\begin{align}\label{eq.2.16}
E=\left[
\begin{array}{ccc}
1&0&0\\
0&1&-\frac{1}{3}\\
0&-\frac{1}{3}&1\\
\end{array}\right],
\end{align}
which is positive definite.
Hence, $E^{-1}$ exists and it is positive definite as well, which implies that we can obtain the uniquely determined positive definite matrix $E^{-\frac{1}{2}}$.
Denote $M_{i}=E^{-\frac{1}{2}}E_{i}E^{-\frac{1}{2}}$ for $i=1,2,\ldots,9$. Then
\begin{equation*}
\sum_{i=1}^{9}M_{i}=E^{-\frac{1}{2}}\Big(\sum_{i=1}^{9}E_{i}\Big)E^{-\frac{1}{2}}=I_{3}.
\end{equation*}
Since $E_{1},E_{2},\ldots,E_{9}$ are linearly independent as elements of $M_{3}(\mathbb{C})$, $M_{1},M_{2},\ldots,M_{9}$ are also linearly independent as elements of $M_{3}(\mathbb{C})$. Therefore, $\{M_{i}|i=1,2,\ldots,9\}$ forms an IC-POVM.

By Eq. (\ref{eq.2.16}), we obtain that
\begin{equation*}
E^{-1}=\left[
\begin{array}{ccc}
1&0&0\\
0&\frac{9}{8}&\frac{3}{8}\\
0&\frac{3}{8}&\frac{9}{8}\\
\end{array}\right]\ \mathrm{and} \ E^{-\frac{1}{2}}=\left[
\begin{array}{ccc}
1&0&0\\
0&\frac{\sqrt{9+6\sqrt{2}}}{4}&\frac{\sqrt{9-6\sqrt{2}}}{4}\\
0&\frac{\sqrt{9-6\sqrt{2}}}{4}&\frac{\sqrt{9+6\sqrt{2}}}{4}\\
\end{array}\right].
\end{equation*}

Denote $\beta:=-\frac{\sqrt{46}(\kappa+\sigma)}{23}$, $\gamma:=\frac{\sqrt{138}(\kappa-\sigma)}{23}$,
$A:=\frac{2\sqrt{46}}{23}(\kappa-\frac{\sigma}{2})$, $B:=\frac{2\sqrt{46}}{23}(\sigma-\frac{\kappa}{2})$ and
$C:=\frac{\sqrt{138}}{23}$, where $\kappa:=\frac{\sqrt{9+6\sqrt{2}}}{4}$ and $\sigma:=\frac{\sqrt{9-6\sqrt{2}}}{4}$.
By Eq. (\ref{eq4.5}), each $E^{-\frac{1}{2}}|v_{i}\rangle$ can be expressed as a multiple of certain normalized vector
(here, for simplicity, we take $z_{i}=1$ for all $i$ in Eq. (\ref{eq4.5})):
\begin{equation*}
E^{-\frac{1}{2}}|v_{1}\rangle=\lambda_{1}|b_{1}\rangle
\triangleq \frac{\sqrt{138}}{12}
\bigg(\frac{2\sqrt{46}}{23},\beta+\gamma i,\beta-\gamma i\bigg)^{T},
\end{equation*}
\begin{equation*}
E^{-\frac{1}{2}}|v_{2}\rangle=\lambda_{2}|b_{2}\rangle
\triangleq \frac{\sqrt{138}}{12}
\bigg(\frac{2\sqrt{46}}{23},\beta-\gamma i,\beta+\gamma i\bigg)^{T},
\end{equation*}
\begin{equation*}
E^{-\frac{1}{2}}|v_{3}\rangle=\lambda_{3}|b_{3}\rangle
\triangleq \frac{2\sqrt{3}}{3}
\bigg(-\frac{1}{4}+\frac{\sqrt{3}}{4}i,
\frac{\sqrt{46}\beta}{8}+\frac{\sqrt{138}\beta}{8}i,
\frac{\sqrt{46}\beta}{8}+\frac{\sqrt{138}\beta}{8}i\bigg)^{T},
\end{equation*}
\begin{equation*}
E^{-\frac{1}{2}}|v_{4}\rangle=\lambda_{4}|b_{4}\rangle
\triangleq \frac{\sqrt{138}}{12}
\bigg(-\frac{\sqrt{3}C}{3}+Ci,A+\sigma C i,B+\kappa C i\bigg)^{T},
\end{equation*}
\begin{equation*}
E^{-\frac{1}{2}}|v_{5}\rangle=\lambda_{5}|b_{5}\rangle
\triangleq \frac{\sqrt{138}}{12}
\bigg(-\frac{\sqrt{3}C}{3}-Ci,A-\sigma C i,B-\kappa C i\bigg)^{T},
\end{equation*}
\begin{equation*}
E^{-\frac{1}{2}}|v_{6}\rangle=\lambda_{6}|b_{6}\rangle
\triangleq \frac{2\sqrt{3}}{3}
\bigg(-\frac{1}{4}-\frac{\sqrt{3}}{4}i,
\frac{\sqrt{46}\beta}{8}-\frac{\sqrt{138}\beta}{8}i,
\frac{\sqrt{46}\beta}{8}-\frac{\sqrt{138}\beta}{8}i\bigg)^{T},
\end{equation*}
\begin{equation*}
E^{-\frac{1}{2}}|v_{7}\rangle=\lambda_{7}|b_{7}\rangle
\triangleq (1,0,0)^{T},
\end{equation*}
\begin{equation*}
E^{-\frac{1}{2}}|v_{8}\rangle=\lambda_{8}|b_{8}\rangle
\triangleq \frac{3\sqrt{2}}{4}\bigg(0,\frac{2\sqrt{2}\kappa}{3},\frac{2\sqrt{2}\sigma}{3}\bigg)^{T},
\end{equation*}
\begin{equation*}
E^{-\frac{1}{2}}|v_{9}\rangle=\lambda_{9}|b_{9}\rangle
\triangleq \frac{3\sqrt{2}}{4}\bigg(0,\frac{2\sqrt{2}\sigma}{3},\frac{2\sqrt{2}\kappa}{3}\bigg)^{T}.
\end{equation*}

For $j=7$, according to the above normalized vectors $\{|b_{i}\rangle\}_{i=1}^{9}$, we can make the choice on
the values of $\{\eta_{i}\}_{i=1}^{9}$ in Eq. (\ref{eq4.8}):
\begin{equation*}
\eta_{1}=\eta_{2}=\eta_{7}=\eta_{8}=\eta_{9}=1,
\end{equation*}
\begin{equation*}
\eta_{3}=\eta_{4}=\frac{1}{2}+\frac{\sqrt{3}}{2}i,
\end{equation*}
\begin{equation*}
\eta_{5}=\eta_{6}=\frac{1}{2}-\frac{\sqrt{3}}{2}i.
\end{equation*}

Based on these results, the corresponding quasi-Householder matrices $\{\widehat{H_{i}}\}_{i=1}^{9}$ from
Eqs. (\ref{eq4.10}) and (\ref{eq4.12}) are then given by
\begin{equation*}
\widehat{H_{1}}=\left[
\begin{array}{ccc}
0.5898&-0.3612-0.4423i&-0.3612+0.4423i\\
-0.3612+0.4423i&0.2051&0.1590+0.7788i\\
-0.3612-0.4423i&0.1590-0.7788i&0.2051\\
\end{array}\right],
\end{equation*}
\begin{equation*}
\widehat{H_{2}}=\left[
\begin{array}{ccc}
0.5898&-0.3612+0.4423i&-0.3612-0.4423i\\
-0.3612-0.4423i&0.2051&0.1590-0.7788i\\
-0.3612+0.4423i&0.1590+0.7788i&0.2051\\
\end{array}\right],
\end{equation*}
\begin{equation*}
\widehat{H_{3}}=\left[
\begin{array}{ccc}
-0.2500+0.4330i&0.6124&0.6124\\
-0.3062-0.5303i&0.3750-0.6495i&-0.1250+0.2165i\\
-0.3062-0.5303i&-0.1250+0.2165i&0.3750-0.6495i\\
\end{array}\right],
\end{equation*}
\begin{equation*}
\widehat{H_{4}}=\left[
\begin{array}{ccc}
-0.2949+0.5108i&-0.3612-0.4423i&-0.3612+0.4423i\\
0.5636+0.0916i&0.3974-0.6884i&0.1946+0.0650i\\
-0.2025+0.5339i&-0.1535-0.1360i&0.3974-0.6884i\\
\end{array}\right],
\end{equation*}
\begin{equation*}
\widehat{H_{5}}=\left[
\begin{array}{ccc}
-0.2949-0.5108i&-0.3612+0.4423i&-0.3612-0.4423i\\
0.5636-0.0916i&0.3974+0.6884i&0.1946-0.0650i\\
-0.2025-0.5339i&-0.1535+0.1360i&0.3974+0.6884i\\
\end{array}\right],
\end{equation*}
\begin{equation*}
\widehat{H_{6}}=\left[
\begin{array}{ccc}
-0.2500-0.4330i&0.6124&0.6124\\
-0.3062+0.5303i&0.3750+0.6495i&-0.1250-0.2165i\\
-0.3062+0.5303i&-0.1250-0.2165i&0.3750+0.6495i\\
\end{array}\right],
\end{equation*}
\begin{equation*}
\widehat{H_{7}}=\left[
\begin{array}{ccc}
1&0&0\\
0&-0.1111&0.5476-0.8293i\\
0&0.5476+0.8293i&0.1111\\
\end{array}\right],
\end{equation*}
\begin{equation*}
\widehat{H_{8}}=\left[
\begin{array}{ccc}
0&0.9856&0.1691\\
0.9856&0.0286&-0.1667\\
0.1691&-0.1667&0.9714\\
\end{array}\right],\ \widehat{H_{9}}=\left[
\begin{array}{ccc}
0&0.1691&0.9856\\
0.1691&0.9714&-0.1667\\
0.9856&-0.1667&0.0286\\
\end{array}\right].
\end{equation*}

Define $H_{i}=\widehat{H_{i}}^{\dag}$, with $\widehat{H_{i}}$ calculated above for $i=1,2,\ldots,9$.
Put $\mu_{i}(t)=\frac{1-\mathrm{e}^{-\theta_{i}t}}{9}$ for $i=1,2,\ldots,8$, and $\mu_{9}(t)=\frac{1+\sum_{s=1}^{8}\mathrm{e}^{-\theta_{s}t}}{9}$,
where $\{\theta_{i}\}_{i=1}^{8}$ represent eight distinct positive decoherence parameters.
By calculating the Wronskian of $\{\mu_{i}(t)\}_{i=1}^{9}$, we deduce that $\{\mu_{i}(t)\}_{i=1}^{9}$ are linearly independent.
Let $t_{1},t_{2},\ldots,t_{9}$ be nine distinct time instants. Then, we know that $\mathrm{det}([\mu_{j}(t_{i})]_{i,j=1}^{9})\neq 0$,
and thus $\mathrm{det}(K)\neq 0$ by Eq. (\ref{eq4.15}).

Now, for any unknown state $\rho(0)\in M_{3}(\mathbb{C})$, we can combine it with the above terms $\{H_{i}\}_{i=1}^{9}$ and $\{\mu_{i}(t)\}_{i=1}^{9}$
to define the following dynamical qudit:
\begin{equation*}
\rho(t)=\sum_{i=1}^{9}\mu_{i}(t)H_{i}\rho(0)H_{i}^{\dag},
\end{equation*}

At this time, once the probabilities $\big\{\mathrm{Prob}\big(t_{i}^{(7)}\big)\big\}_{i=1}^{9}$ of outcome $7$ measured at nine distinct time instants
$t_{1},t_{2},\ldots,t_{9}$ are experimentally known, the values of the set $\{\mathrm{tr}(M_{i}\rho(0))\}_{i=1}^{9}$ are determined, and thus the complete knowledge about the unknown state $\rho(0)$ is extracted.
\end{example}

\section{Noise could be useful in quantum state tomography}\label{section3}

In Section \ref{section2}, we showed a dynamical tomography method for the initial state of a evolving system.
By choosing an IC-POVM, the method works for certain corresponding dynamical process.
In this section, we will consider the opposite question: which families of IC-POVMs can be obtained from certain dynamical process?
We will show that under the time-dependent average channel $\Psi_t(\rho)$ defined in Eq. (\ref{eq5.5}), we can acquire a collection of projective operators that is informationally complete (IC) and thus obtain the corresponding IC-POVM (see Theorem \ref{theorem5.1}). Moreover, we also show that under certain condition, it is possible to acquire infinite families of projective operators that are IC, and obtain infinite families of corresponding IC-POVMs (see Theorem \ref{theorem5.4} and Corollary \ref{corollary5.5}); otherwise, the Zauner's conjecture is incorrect.

We show that for any unknown quantum state $\rho_{0}$ which is evolved with the time-dependent average channel $\Psi_t(\rho)$ defined in Eq. (\ref{eq5.5}), only one projective measurement $\{|\phi\rangle\langle\phi|,I-|\phi\rangle\langle\phi|\}$ in $d^2$ different time instants could be enough to reconstruct $\rho_{0}$ (see Eq. (\ref{eq5.7})), instead of the common $d^2$ different measurements. As an application, we show that we can simulate a SIC-POVM on any unknown quantum state by using the time-dependent average channel $\Psi_t(\rho)$ (see Claim \ref{claim5.7}).

\subsection{The time-dependent average channel}\label{subsection5.1}

Usually, noise is not what we want in quantum information processing.
Noise effects the evolution of the main system and makes it no longer unitary which can be expressed as the operator-sum form
$\epsilon(\rho)=\sum_{k} E_k \rho E_k^{\dag}$.
Perhaps one of the simplest noise models of a quantum system is the depolarizing channel \cite{Burrell2009Geometry,King2003The}.
If this channel is time-dependent acting on the initial state $\rho_0$, we could express it as
\begin{equation}\label{eq5.1}
\epsilon_0(\rho)=\lambda_{0}(t)\rho_0+\frac{1-\lambda_{0}(t)}{d} I,
\end{equation}
where $0\le \lambda_{0}(t)\le 1$.

The set of Weyl-Heisenberg bases is a possible way to generalize the Pauli bases.
Let us rewrite the depolarizing channel through this set.

For an arbitrary $d$-dimensional space, define
\begin{equation*}
	X=\sum_{j=0}^{d-1}|j+1\rangle\langle j|,~~ Z=\sum_{j=0}^{d-1}\omega^j|j\rangle\langle j|,
\end{equation*}
where $\omega=e^{2\pi i/d}$.
It is not difficult to verify that $X^r=\sum_{j=0}^{d-1}|j+r\rangle\langle j|$, $Z^r=\sum_{j=0}^{d-1}\omega^{jr}|j\rangle\langle j|$.

Now we may define the Weyl-Heisenberg bases as

\begin{equation}\label{eq5.2}
M_{\alpha}:=M_{jk}=X^jZ^k,
\end{equation}
where $j,k\in \{0,\cdots,d-1\}$ and $\alpha=jd+k$.

With some simple calculation, $X^jZ^k=\sum\omega^{tk}|t+j\rangle\langle t|$, $Z^kX^j=\omega^{jk}\sum\omega^{tk}|t+j\rangle\langle t|$. Then
\begin{equation*}
X^jZ^k=\omega^{-jk}Z^kX^j.
\end{equation*}
By \cite{Burrell2009Geometry}, we have
\begin{enumerate}
	\item $M_0=I$,
	\item 	$\mbox{tr}(M_{\alpha})=0$ for all $\alpha\ne 0$,
	\item  $\mbox{tr}(M_{\alpha}^{\dag}M_{\beta})=0$ for all $\alpha\ne \beta$.
\end{enumerate}

Then for any $d$-dimensional quantum state $\rho$,
\begin{equation*}
	\rho=\frac{1}{d} \sum_{\alpha=0}^{d^2-1}\mbox{tr}(M_{\alpha}^{\dag}\rho)M_{\alpha}
	=\frac{1}{d} \sum_{j,k=0}^{d-1}\mbox{tr}(M_{jk}^{\dag}\rho)M_{jk}.
\end{equation*}

Taking Eq. (\ref{eq5.2}) into consideration, we have

\begin{equation*}
M_{mn}\rho M_{mn}^{\dag}=\frac{1}{d}\sum_{j,k=0}^{d-1}\omega^{km+nj-2jk}\mbox{tr}(M_{jk}^{\dag}\rho)M_{jk}.
\end{equation*}

For any $(j,k)\ne (0,0)$, $\sum_{m,n=0}^{d-1}\omega^{km+nj-2jk}=0$.
For $j=k=0$, $\sum_{m,n=0}^{d-1}\omega^{km+nj-2jk}=d^2$.
Then we have
\begin{equation*}
	\sum_{m,n=0}^{d-1}M_{mn}\rho M_{mn}^{\dag}=d I.
\end{equation*}
Hence, the time-dependent depolarizing channel in Eq. (\ref{eq5.1}) is rewritten as
\begin{equation*}
	\epsilon_0(\rho)=\lambda_{0}(t)\rho_0+\frac{1-\lambda_{0}(t)}{d^2} \sum_{m,n=0}^{d-1}M_{mn}\rho_0 M_{mn}^{\dag}.
\end{equation*}

Define another type of channel as
\begin{equation*}
	\epsilon_k(\rho)=\lambda_k(t)\rho_0+(1-\lambda_k(t))M_{k}\rho_0 M_{k}^{\dag},~k=1,\cdots,d^2-1.
\end{equation*}

If we evolve the initial state $\rho_0$ with the above $d^{2}$ channels with probability $1/d^2$, then a new channel called the \textbf{time-dependent average channel} is defined as
\begin{align}\label{eq5.3}
\Psi_t(\rho)
&=\frac{\sum_{i=0}^{d^{2}-1}\epsilon_i(\rho)}{d^{2}}\nonumber\\
&=\frac{1-\lambda_{0}(t)+d^{2}\sum_{i=0}^{d^{2}-1}\lambda_{i}(t)}{d^{4}}\rho_{0}
+\sum_{i=1}^{d^{2}-1}\frac{1-\lambda_{0}(t)+d^{2}(1-\lambda_{i}(t))}{d^{4}}M_{i}\rho_{0} M_{i}^{\dag}.
\end{align}

\subsection{Which families of IC-POVMs can be obtained from the time-dependent average channel?}\label{subsection5.2}

We can select a discrete number of different time instants $\{t_1,\cdots,t_{d^2}\}$. We measure the evolved state in these different time instants with a POVM $\{E_0=|\phi\rangle\langle \phi|,I-E_0\}$. Denote the probability as $p(t)$ if the state $\Psi_t(\rho)$ is collapsed into $|\phi\rangle\langle \phi|$.
Then the probability $p(t)$ can be expressed as
\begin{equation}\label{eq5.4}
p(t)=\mbox{tr}[\Psi_t(\rho)|\phi\rangle\langle \phi|].
\end{equation}

By Eq. (\ref{eq5.3}), we can rewrite the time-dependent average channel $\Psi_t(\rho)$ as
\begin{equation}\label{eq5.5}
\Psi_t(\rho)= \sum_{i=0}^{d^{2}-1} \mu _{i}(t) M_{i}\rho_0 M_{i}^{\dag},
\end{equation}
where $\rho_0$ is an unknown quantum state,
$\mu _{0}(t)=\frac{1-\lambda_{0}(t)+d^{2}\sum_{i=0}^{d^{2}-1}\lambda_{i}(t)}{d^{4}}$ and
$\mu_{i}(t)=\frac{1-\lambda_{0}(t)+d^{2}(1-\lambda_{i}(t))}{d^{4}}$ for $i=1,\ldots,d^{2}-1$.
Substituting Eq. (\ref{eq5.5}) into Eq. (\ref{eq5.4}), we have
\begin{equation}\label{eq5.6}
p(t)=\mbox{tr}\left[\sum_{i=0}^{d^{2}-1} \mu _{i}(t) M_{i}\rho_0 M_{i}^{\dag}|\phi\rangle\langle \phi|\right]
= \sum_{i=0}^{d^{2}-1} \mu _{i}(t) \mbox{tr}[\rho_0 M_{i}^{\dag}|\phi\rangle\langle \phi|M_{i}].
\end{equation}

Now, for $d^2$ different time instants $\{t_1,\cdots,t_{d^2}\}$, it follows from Eq. (\ref{eq5.6}) that
\begin{equation}\label{eq5.7}
\left[
\begin{matrix}
	p(t_1)\\
	p(t_2)\\
	\vdots\\
	p(t_{d^2})
\end{matrix}
\right]= \left[
\begin{matrix}
	\mu_{0}(t_1) &\mu_{1}(t_1) &\cdots &\mu_{d^{2}-1}(t_1) \\
	\mu_{0}(t_2) &\mu_{1}(t_2)  &\cdots &\mu_{d^{2}-1}(t_2)\\
	\vdots        & 	\vdots   &\ddots &\vdots\\
	\mu_{0}(t_{d^2}) &\mu_{1}(t_{d^2})  &\cdots &\mu_{d^{2}-1}(t_{d^2})
\end{matrix}
\right] \left[
\begin{matrix}
	\mbox{tr}[\rho_0 M_{0}^{\dag}|\phi\rangle\langle \phi|M_{0}]\\
	\mbox{tr}[\rho_0 M_{1}^{\dag}|\phi\rangle\langle \phi|M_{1}]\\
	\vdots\\
	\mbox{tr}[\rho_0 M_{d^{2}-1}^{\dag}|\phi\rangle\langle \phi|M_{d^{2}-1}]
\end{matrix}
\right].
\end{equation}

We use $\mathcal{U}$ to denote the square matrix in Eq. (\ref{eq5.7}). Then it is not difficult to verify that $\mathcal{U}$ is invertible.
Therefore, we can uniquely obtain the probability distribution $\big\{\mbox{tr}[\rho_0 M_{\alpha}^{\dag}|\phi\rangle\langle \phi|M_{\alpha}]:\alpha=0,\cdots,d^{2}-1\big\}$.
This reveals that for any unknown quantum state $\rho_{0}$ which is evolved with the time-dependent average channel $\Psi_t(\rho)$ defined in Eq. (\ref{eq5.5}), only one projective measurement $\{|\phi\rangle\langle\phi|,I-|\phi\rangle\langle\phi|\}$ in $d^2$ different time instants could be enough to reconstruct $\rho_{0}$, instead of the common $d^2$ different measurements.

Denote $|\phi_{\alpha}\rangle:=M_{\alpha}|\phi\rangle$, where $\alpha=0,\cdots,d^2-1$.
Observing Eq. (\ref{eq5.7}), we know that the initial state $\rho_0$ is uniquely determined as along as $\{|\phi_{\alpha}\rangle \langle \phi_{\alpha}| :\alpha=0,\cdots,d^2-1\}$ could span the whole $d$-dimensional density matrix space, i.e., informationally complete (IC).

Based on Eq. (\ref{eq5.7}), we will give a sufficient condition under which a collection of projective operators is IC
and obtain the corresponding IC-POVM in the following theorem.

\begin{theorem}\label{theorem5.1}
Define a real symmetric square matrix $\mathcal{A}:=\left[|\langle \phi_j|\phi_k\rangle|^2\right]_{0\leq j,k\leq d^{2}-1}$.
Let $K=\sum_{\alpha=0}^{d^{2}-1}|\phi_{\alpha}\rangle \langle \phi_{\alpha}|$.
Suppose that $|\mathcal{A}| \ne 0$. Then

(1) $\{|\phi_{\alpha}\rangle \langle \phi_{\alpha}| :\alpha=0,\cdots,d^2-1\}$ is informationally complete (IC);

(2) $\{K^{-\frac{1}{2}}|\phi_{\alpha}\rangle \langle \phi_{\alpha}|K^{-\frac{1}{2}} :\alpha=0,\cdots,d^2-1\}$ is an IC-POVM.
\end{theorem}

\begin{proof} (1) Suppose there are some complex numbers $x_{0},\ldots,x_{d^{2}-1}$ such that
$\sum_{\alpha=0}^{d^2-1} x_{\alpha} |\phi_{\alpha}\rangle\langle \phi_\alpha|=0$.
Then for any $|\phi_{\beta}\rangle$, we have $\langle \phi_\beta|\left(\sum_{\alpha=0}^{d^2-1} x_{\alpha} |\phi_{\alpha}\rangle\langle \phi_\alpha|\right)|\phi_\beta\rangle=0$, i.e.,
$\sum_{\alpha=0}^{d^2-1} x_{\alpha} |\langle \phi_{\beta}|\phi_{\alpha}\rangle|^2 =0$. That is,
\begin{equation*}
\mathcal{A}|x\rangle=\left[
\begin{matrix}
	|\langle \phi_0|\phi_0\rangle|^2   & |\langle \phi_0|\phi_1\rangle|^2   &\cdots    &|\langle \phi_0|\phi_{d^2-1}\rangle|^2 \\
	|\langle \phi_1|\phi_0\rangle|^2   & |\langle \phi_1|\phi_1\rangle|^2   &\cdots    &|\langle \phi_1|\phi_{d^2-1}\rangle|^2\\
	\vdots           &  \vdots         &\ddots&\vdots\\
	|\langle \phi_{d^2-1}|\phi_0\rangle|^2   & |\langle \phi_{d^2-1}|\phi_1\rangle|^2   &\cdots    &|\langle \phi_{d^2-1}|\phi_{d^2-1}\rangle|^2
\end{matrix}
\right]
\left[
\begin{matrix}
	x_0\\
	x_1\\
	\vdots\\
	x_{d^2-1}
\end{matrix}
\right]=0.
\end{equation*}
Since $|\mathcal{A}| \ne 0$, we have $x_{0}=\cdots=x_{d^{2}-1}=0$.
Hence $\{|\phi_{\alpha}\rangle \langle \phi_{\alpha}| :\alpha=0,\cdots,d^2-1\}$ is linear independent and it thus spans the whole density matrix space.
Therefore, $\{|\phi_{\alpha}\rangle \langle \phi_{\alpha}| :\alpha=0,\cdots,d^2-1\}$ is IC.

(2) By part (1), it is not difficult to verify that $K$ is positive definite and so is $K^{-1}$, which implies that we can obtain the uniquely determined positive definite matrix $K^{-\frac{1}{2}}$, and thus
\begin{equation*}
\sum_{\alpha=0}^{d^{2}-1}K^{-\frac{1}{2}}|\phi_{\alpha}\rangle \langle \phi_{\alpha}|K^{-\frac{1}{2}}=K^{-\frac{1}{2}}\Big(\sum_{\alpha=0}^{d^{2}-1}|\phi_{\alpha}\rangle \langle \phi_{\alpha}|\Big)K^{-\frac{1}{2}}=I_{d}.
\end{equation*}
Therefore, $\{K^{-\frac{1}{2}}|\phi_{\alpha}\rangle \langle \phi_{\alpha}|K^{-\frac{1}{2}} :\alpha=0,\cdots,d^2-1\}$ is an IC-POVM. $\hfill\square$
\end{proof}

\begin{corollary}\label{corollary5.2}
\rm{(\!\!\cite[pp.2175-2176]{Renes2004Symmetric})}
If there is a state $|\phi\rangle$ such that
\begin{equation}\label{eq5.8}
	|\langle \phi | M_{\alpha}|\phi\rangle|^2= \frac{1}{d+1}, ~\alpha=1,\cdots,d^2-1,
\end{equation}
then $\{\frac{1}{d}|\phi_{\alpha}\rangle \langle \phi_{\alpha}| :\alpha=0,\cdots,d^2-1\}$ with $|\phi_{\alpha}\rangle=M_{\alpha}|\phi\rangle$
forms a SIC-POVM in dimension $d$.
Here, the state $|\phi\rangle$ satisfying Eq. (\ref{eq5.8}) is called a \textbf{fiducial state}.
\end{corollary}

\begin{proof}
By Eq. (\ref{eq5.8}), the matrix $\mathcal{A}:=\left[|\langle \phi_j|\phi_k\rangle|^2\right]_{0\leq j,k\leq d^{2}-1}$ with $|\phi_{\alpha}\rangle=M_{\alpha}|\phi\rangle$ is given by
\begin{equation*}
\mathcal{A}=\frac{1}{d+1}\left[
\begin{matrix}
	d+1   & 1   &\cdots    &1 \\
	1   & d+1   &\cdots    &1\\
	\vdots           &\vdots           &\ddots&\vdots\\
	1   & 1   &\cdots    &d+1
\end{matrix}
\right].
\end{equation*}
It is not difficult to verify that $|\mathcal{A}|> 0$. By Theorem \ref{theorem5.1}, we deduce that
$\{\frac{1}{d}|\phi_{\alpha}\rangle \langle \phi_{\alpha}| :\alpha=0,\cdots,d^2-1\}$ is IC.
Besides, by direct calculation, we have $\sum_{\alpha=0}^{d^{2}-1}\frac{1}{d}|\phi_{\alpha}\rangle \langle \phi_{\alpha}|=I_{d}$,
which together with the fact that $\{\frac{1}{d}|\phi_{\alpha}\rangle \langle \phi_{\alpha}| :\alpha=0,\cdots,d^2-1\}$ is symmetric reveals that
$\{\frac{1}{d}|\phi_{\alpha}\rangle \langle \phi_{\alpha}| :\alpha=0,\cdots,d^2-1\}$ is a SIC-POVM. $\hfill\square$
\end{proof}

\begin{remark}
Corollary \ref{corollary5.2} reveals that the problem of constructing a SIC-POVM can be transformed into that of finding a fiducial state $|\phi\rangle$
satisfying Eq. (\ref{eq5.8}).
Such SIC-POVMs obtained from the fiducial states are called the \textbf{Weyl-Heisenberg covariant SIC-POVMs}.
It was conjectured in \cite{Renes2004Symmetric} that in every finite dimension, a Weyl-Heisenberg covariant SIC-POVM can be constructed
as the orbit of a suitable fiducial state $|\phi\rangle$ satisfying Eq. (\ref{eq5.8}).
In fact, such fiducial states do exist for some dimensions $d$.
For example,
when $d=2$, there exists a fiducial state $|\phi\rangle=\frac{1}{\sqrt{6}}\left[\sqrt{3+\sqrt{3}},e^{i\pi/4}\sqrt{3-\sqrt{3}}\right]^{T}$ satisfying Eq. (\ref{eq5.8});
when $d=3$, there exists a fiducial state $|\phi\rangle=\frac{1}{\sqrt{2}}[0,1,-1]^{T}$ satisfying Eq. (\ref{eq5.8}).
Besides, such Weyl-Heisenberg covariant SIC-POVMs were found with high numerical precision in all dimensions $d\leq 45$.
In \cite{Zauner1999Quantum}, Zauner put forward a stronger conjecture.
He conjectured that in all finite dimensions there exists a fiducial state for a Weyl-Heisenberg covariant SIC-POVM that is an eigenvector of the matrix $Z$,
where $Z$ satisfies $\langle j|Z|k\rangle:=\frac{e^{i\xi}}{\sqrt{d}}\tau^{2jk+j^{2}}$ in which $\xi\in \mathbb{R}$ and
$\tau=e^{\frac{\pi i(d+1)}{d}}$.
Recently, Horodecki, Rudnicki and Zyczkowski \cite{Horodecki2022Five} listed the problem whether the SIC-POVMs can be constructed in an infinite sequence
of dimensions as the first open problem in quantum information theory.
This problem is also related to the maximal set of equiangular line and Hilbert's 12th problem.
\end{remark}

In the following theorem, we show that under certain condition, it is possible to obtain infinite families of projective operators that are IC,
and obtain infinite families of corresponding IC-POVMs.

\begin{theorem}\label{theorem5.4}
Define $\mathcal{A}_{|\phi\rangle}:=\left[|\langle \phi_j|\phi_k\rangle|^2\right]_{0\leq j,k\leq d^{2}-1}$,
where $|\phi_{\alpha}\rangle=M_{\alpha}|\phi\rangle$ for $\alpha=0,\cdots,d^{2}-1$.
If $|\mathcal{A}_{|\phi\rangle}|\ne 0$, then there exist infinite states $|\widehat{\phi}\rangle$ such that

(1) $|\mathcal{A}_{|\widehat{\phi}\rangle}|\ne 0$;

(2) $\{|\widehat{\phi_{\alpha}}\rangle \langle \widehat{\phi_{\alpha}}| :\alpha=0,\cdots,d^2-1\}$ is IC, where $|\widehat{\phi_{\alpha}}\rangle=M_{\alpha}|\widehat{\phi}\rangle$;

(3) $\{K^{-\frac{1}{2}}|\widehat{\phi_{\alpha}}\rangle \langle \widehat{\phi_{\alpha}}|K^{-\frac{1}{2}} :\alpha=0,\cdots,d^2-1\}$ is an IC-POVM,
where $K=\sum_{\alpha=0}^{d^{2}-1}|\widehat{\phi_{\alpha}}\rangle \langle \widehat{\phi_{\alpha}}|$.
\end{theorem}

\begin{proof}
(1) Define a vector $|\widetilde{\phi}\rangle=\sum_{k=0}^{d-1} (a_k+i b_k)|k\rangle$ with real variables $\{a_k,b_k:k=0,\cdots,d-1\}$.
Set $M_j=M_{j_1j_{2}}=X^{j_1}Z^{j_2}$ and $M_k=M_{k_1k_2}=X^{k_1}Z^{k_2}$.
Then the $(j,k)$-th element of $\mathcal{A}_{|\widetilde{\phi}\rangle}$ is
\begin{align*}
\left|\langle \widetilde{\phi_{j}}|\widetilde{\phi_{k}}\rangle\right|^2
		            &=\left|\langle \widetilde{\phi}| Z^{-j_2}X^{-j_1}X^{k_1}Z^{k_2}|\widetilde{\phi} \rangle\right|^2\\
		            &=\left|\langle \widetilde{\phi}| X^{k_1-j_1}Z^{k_2-j_2}|\widetilde{\phi} \rangle\right|^2\\
		            &=\Bigg|\sum_{k,j=0}^{d-1} (a_k-i b_k)(a_j+i b_j)\omega^{j(k_2-j_2)}\langle k|j+k_1-j_1\rangle\Bigg|^2\\
		            &=\Bigg|\sum_{k,j=0}^{d-1} \left\{\cos [j(k_2-j_2)](a_ka_j+b_kb_j)-\sin     [j(k_2-j_2)](a_kb_j- b_ka_j)\right\} \delta_{k,j+k_1-j_1}\Bigg|^2\\
		            &+\Bigg|\sum_{k,j=0}^{d-1}\left\{\sin [j(k_2-j_2)](a_ka_j+b_kb_j)+\cos     [j(k_2-j_2)](a_kb_j- b_ka_j)\}\delta_{k,j+k_1-j_1}\right\}\Bigg|^2,
\end{align*}
which reveals that $\left|\langle \widetilde{\phi_{j}}|\widetilde{\phi_{k}}\rangle\right|^2$ is a polynomial with real variables $\{a_k,b_k:k=0,\cdots,d-1\}$,
and so is $|\mathcal{A}_{|\widetilde{\phi}\rangle}|$.
Denote $|\mathcal{A}_{|\widetilde{\phi}\rangle}|:=f(a_0,\cdots,a_{d-1},b_0,\cdots,b_{d-1})$.
Since $|\mathcal{A}_{|\phi\rangle}|\ne 0$, we have that
$$|\mathcal{A}_{|\widetilde{\phi}\rangle}|=f(a_0,\cdots,a_{d-1},b_0,\cdots,b_{d-1})\not\equiv0,$$
which implies that there are infinite vectors $|\widetilde{\phi}\rangle$ such that $|\mathcal{A}_{|\widetilde{\phi}\rangle}|\ne 0$.
Making a normalization to these vectors $|\widetilde{\phi}\rangle$, we are able to obtain infinite corresponding states $|\widehat{\phi}\rangle$
such that $|\mathcal{A}_{|\widehat{\phi}\rangle}|\ne 0$.

(2) It is directly obtained from part (1) and Theorem \ref{theorem5.1} (1).

(3) Similar to the proof of Theorem \ref{theorem5.1} (2). $\hfill\square$
\end{proof}

\vspace{6pt}

The following corollary derived from Theorem \ref{theorem5.4} reveals that we can acquire infinite families of projective operators that are IC
and obtain infinite families of corresponding IC-POVMs, provided that there exists a fiducial state $|\phi\rangle$ satisfying Eq. (\ref{eq5.8}).

\begin{corollary}\label{corollary5.5}
If there exists a fiducial state $|\phi\rangle$ satisfying Eq. (\ref{eq5.8}), then there exist infinite states $|\widehat{\phi}\rangle$ such that

(1) $|\mathcal{A}_{|\widehat{\phi}\rangle}|\ne 0$;

(2) $\{|\widehat{\phi_{\alpha}}\rangle \langle \widehat{\phi_{\alpha}}| :\alpha=0,\cdots,d^2-1\}$ is IC, where $|\widehat{\phi_{\alpha}}\rangle=M_{\alpha}|\widehat{\phi}\rangle$;

(3) $\{K^{-\frac{1}{2}}|\widehat{\phi_{\alpha}}\rangle \langle \widehat{\phi_{\alpha}}|K^{-\frac{1}{2}} :\alpha=0,\cdots,d^2-1\}$ is an IC-POVM,
where $K=\sum_{\alpha=0}^{d^{2}-1}|\widehat{\phi_{\alpha}}\rangle \langle \widehat{\phi_{\alpha}}|$.
\end{corollary}

\begin{proof}
By the proof of Corollary \ref{corollary5.2}, we know
$|\mathcal{A}_{|\phi\rangle}|>0$ for $\mathcal{A}_{|\phi\rangle}:=\left[|\langle \phi_j|\phi_k\rangle|^2\right]_{0\leq j,k\leq d^{2}-1}$ with $|\phi_{\alpha}\rangle=M_{\alpha}|\phi\rangle$. Then the result follows from Theorem \ref{theorem5.4}. $\hfill\square$
\end{proof}

\begin{remark}
On Corollary \ref{corollary5.5}, we make the following remarks:

1) Corollary \ref{corollary5.5} implies that if there is a fiducial state $|\phi\rangle$ satisfying Eq. (\ref{eq5.8}),
then we will have more space to design the unitary operation to prepare the state $|\widehat{\phi}\rangle$ from $|0\rangle$.

2) If $|\mathcal{A}_{|\widetilde{\phi}\rangle}|=f(a_0,\cdots,a_{d-1},b_0,\cdots,b_{d-1})\equiv0$, then
we deduce from Corollary \ref{corollary5.5} and the proof of Theorem \ref{theorem5.4} that there is no fiducial state $|\phi\rangle$ satisfying Eq. (\ref{eq5.8}).
This implies that the Zauner's conjecture is incorrect in this case.
\end{remark}

\subsection{An application of the time-dependent average channel}\label{subsection5.3}

Given a quantum state $\rho_0$, can we simulate a SIC-POVM $\{E_k:k=1,\cdots,d^2\}$ on it by using the time-dependent average channel shown in Eq. (\ref{eq5.5})?
The answer is yes. We put this statement in the following claim.

\begin{claim}\label{claim5.7}
We can simulate a SIC-POVM on any unknown quantum state by using the time-dependent average channel in Eq. (\ref{eq5.5}).
\end{claim}

\begin{proof}
For any unknown quantum state $\rho_0$, we can simulate a SIC-POVM $\{E_k:k=0,\cdots,d^{2}-1\}$ on it by the following steps:

{\bfseries Step 1:} We firstly evolve $\rho_0$ with the time-dependent average channel shown in Eq. (\ref{eq5.5}).

{\bfseries Step 2:} Then, we measure the evolved state with $\{|\phi\rangle\langle \phi|,I-|\phi\rangle\langle \phi|\}$ at $d^2$ different time instants.
By Eq. (\ref{eq5.7}), the probabilities $\{\mbox{tr}( \rho_0 M_k|\phi\rangle \langle \phi|M_k^{\dag}):k=0,\cdots,d^{2}-1\}$ can be calculated.

{\bfseries Step 3:} By \cite{Zauner1999Quantum,Fuchs2017The}, when $|\phi\rangle$ is the fiducial state satisfying Eq. (\ref{eq5.8}),
we have $M_k|\phi\rangle \langle \phi|M_k^{\dag}=|\phi_k\rangle\langle\phi_k|=d E_k$.

{\bfseries Step 4:} Dividing a factor $d$ by each of these probabilities calculated in Step 2, we obtain the probability distributions
$\{\mbox{tr}(\rho_{0} E_k):k=0,\cdots,d^{2}-1\}$.

Therefore, the probabilities of $\rho_{0}$ measured by a SIC-POVM $\{E_k:k=0,\cdots,d^{2}-1\}$ are simulated. $\hfill\square$
\end{proof}

\begin{table}[!htbp]
\renewcommand\arraystretch{1.8}
\centering	
\small
\setlength{\abovecaptionskip}{0.cm}
\setlength{\belowcaptionskip}{0.2cm}
\caption{A summary of some known methods to implement SIC-POVMs experimentally.}\label{table1}
\vspace{6pt}
\begin{tabular}{|c|c|c|}
\hline
No.&Method (Main tools)&Reference\\
\hline
(a)&Neumark's theorem&\cite{Rehacek2004Minimal,Tabia2012Experimental}\\
\hline
(b)&Quantum walk&\cite{Bian2015Realization,Zhao2015Experimental}\\
\hline
(c)&Unitary operations on the fiducial state without auxiliary space&\cite{Bent2015Experimental}\\
\hline
(d)&Projective measurements and post-processing procedure&\cite{Singal2022Implementation}\\
\hline
\end{tabular}

\end{table}

\begin{remark}
In addition to full state tomography, the SIC-POVMs can be used in some other areas of quantum information, such as self-testing \cite{Tavakoli2021Mutually}, quantum cryptography \cite{Matthews2009Distinguishability}, entanglement detection \cite{Shang2018Enhanced} and etc.
Some different methods to implement the SIC-POVMs experimentally are summarized in Table \ref{table1}. To be specific,
\begin{itemize}
\setlength{\itemsep}{1pt}
\setlength{\parsep}{1pt}
\setlength{\parskip}{1pt}
\item [(a)] The first method is to use Neumark's theorem. For the SIC-POVM of dimension 2, two \cite{Rehacek2004Minimal} or one \cite{Tabia2012Experimental}
auxiliary qubits have been introduced. After preparing  a global unitary operation between the main system and the auxiliary qubits, people  measure the compound system with standard projective measurements.

\item [(b)] The second method is to use quantum walk, where one dimensional line is the auxiliary space \cite{Bian2015Realization,Zhao2015Experimental}. People also have to prepare the global unitary operation, but the coupling operations between the main system and auxiliary space are different from the ones in qubit systems.
In some specific physical platform, the operations can be implemented easily.

\item [(c)] The third method is to use $d^2$ unitary operations on the fiducial state without auxiliary space. Tomography of qubits in dimension $d=2-10$ has been reported in \cite{Bent2015Experimental}.

\item [(d)] The fourth method is to simulate SIC-POVMs probabilistically \cite{Singal2022Implementation}.
People use the projective measurements and post-processing procedure to simulate the generalized POVMs.
The SIC-POVMs are implemented in more dimensions with the success probability about $\frac{1}{5}$.
\end{itemize}
\end{remark}

\section{Concluding remarks and further research}\label{section4}
The main contributions in this paper are summarized as follows.
\begin{itemize}
\setlength{\itemsep}{1.5pt}
\setlength{\parsep}{1.5pt}
\setlength{\parskip}{1.5pt}
\item [(1)] We established a dynamical quantum state tomography framework.
Under this framework, it is feasible to obtain complete knowledge of any unknown state of a $d$-level system via only an arbitrary operator of certain
types of IC-POVMs in dimension $d$ (see Theorem \ref{theorem4.4}).
We also provided a concrete example to illustrate this framework (see Example \ref{example4.8}).

\item [(2)] We studied the IC-POVMs under the time-dependent average channel
and show how to simulate a SIC-POVM on any unknown quantum state by using this channel.
\begin{itemize}
\item We showed that under the time-dependent average channel, we can acquire a collection of projective operators that is informationally complete (IC) and thus obtain the corresponding IC-POVM (see Theorem \ref{theorem5.1}).

\item We showed that under certain condition, it is possible to obtain infinite families of projective operators that are IC, and obtain infinite families of corresponding IC-POVMs (see Theorem \ref{theorem5.4} and Corollary \ref{corollary5.5});
    otherwise, the Zauner's conjecture is incorrect.

\item As an application, we showed that we can simulate a SIC-POVM on any unknown quantum state by using the time-dependent average channel (see Claim \ref{claim5.7}).
\end{itemize}
\end{itemize}

In future research, it would be interesting to explore whether there is a dimension $d$ satisfying that
$\mathcal{A}_{|\phi\rangle}=0$ with numerical experiments, where $\mathcal{A}_{|\phi\rangle}$ is defined as in Theorem \ref{theorem5.4}. If we find such a dimension, it would challenge Zauner's conjecture,
but this doesn't deny that SIC-POVMs exist for all dimensions. Instead of using the Weyl-Heisenberg group, there could also be other different ways to create SIC-POVMs.
Another interesting research topic is to find efficient ways to decompose quantum circuits $U$ to make $\mathcal{A}_{U|0\rangle}\ne 0$, or use variational quantum circuits to create the fiducial state $|\phi\rangle=U|0\rangle$. These methods could be practical for experiments.

Finding an auxiliary system and coupling it with the $d$-dimensional system could create the desired time-dependent channel by performing a partial trace on the ancilla. In real experiments, coefficients in time-dependent depolarizing channels might all be different. Once we figure out these coefficients, they can be directly applied in the dynamical tomography scheme.

\section*{Acknowledge}

Meng Cao and Yu Wang are supported by grants from Yanqi Lake Beijing Institute of Mathematical Sciences and Applications.
Yu Wang is also supported by National Natural Science Foundation of China (Grant No. 62001260).

\footnotesize{
\bibliographystyle{plain}
\bibliography{tomography}
}

\end{document}